\documentclass[twoside]{article}

\usepackage[dvips]{graphicx}
\usepackage[utf8]{inputenc}
\usepackage[english,russian]{babel}
\usepackage{amsmath,amsthm,amssymb}
\usepackage{euscript}

\PassOptionsToPackage{lowtilde}{url}
\RequirePackage[%
    colorlinks=true,
    pdfstartview=FitH,
    linkcolor=blue,
    citecolor=blue,
    filecolor=blue,
    urlcolor=blue,
    linktoc=page,
    pdftitle={On Gate Complexity of Reversible Circuits Consisting of NOT, CNOT and 2-CNOT Gates},
    pdfauthor={Dmitry V. Zakablukov},
    pdfkeywords={reversible logic, gate complexity, computations with memory},
    pdfsubject={The paper discusses the gate complexity of reversible circuits consisting of NOT, CNOT and 2-CNOT gates}]{hyperref}

\newcommand{\mcirc}{\mathop{\circ}\limits}
\newcommand{\ZZ}{\mathbb Z}
\newcommand{\frS}{\mathfrak S}
\newcommand{\vv}{\mathbf}

\title{О сложности обратимых схем, состоящих из функциональных элементов NOT, CNOT и 2-CNOT}
\author{Д.\,В.~Закаблуков}

\theoremstyle{plain}
\newtheorem{theorem}{Теорема}
\newtheorem{lemma}{Лемма}

\newtheorem{assertion}{Утверждение}

\theoremstyle{definition}
\newtheorem{definition}{Определение}

\begin{document}

\maketitle

\begin{abstract}
В работе рассматривается вопрос сложности обратимых схем, состоящих из функциональных элементов NOT, CNOT и 2-CNOT.
Определяется функция Шеннонa $L(n, q)$ сложности обратимой схемы, реализующей отображение $f\colon \ZZ_2^n \to \ZZ_2^n$,
как функция от $n$ и количества дополнительных входов схемы $q$.
Доказывается общая нижняя оценка сложности обратимой схемы $L(n,q) \geqslant \frac{2^n(n-2)}{3\log_2(n+q)} - \frac{n}{3}$.
Доказывается верхняя оценка сложности $L(n,0) \leqslant 3n2^{n+4}(1+o(1)) \mathop / \log_2n$
в случае отсутствия дополнительных входов.
Доказывается асимптотическая верхняя оценка сложности $L(n,q_0) \lesssim 2^n$
в случае использования $q_0 \sim n2^{n-o(n)}$ дополнительных входов.
\end{abstract}

\textbf{Ключевые слова}: обратимые схемы, сложность схемы, вычисления с памятью.

\section*{Введение}
В дискретной математике нередко возникает задача оценить сложность того или иного преобразования.
Теория схемной сложности берет свое начало с работы Шеннона~\cite{shannon}. В ней он предложил в качестве меры сложности
булевой функции рассматривать сложность реализующей ее минимальной контактной схемы.
На сегодняшний день известна асимптотическая оценка сложности $L(n) \sim 2^n \mathop / n$
булевой функции~\cite{yablonsky} в базисе классических функциональных элементов <<инвертор, дизъюнктор,
конъюнктор>>.

В работе~\cite{karpova} рассматривается вопрос о вычислениях с ограниченной памятью. Было доказано, что в базисе
всех $p$-местных булевых функций нижняя асимптотическая оценка сложности схемы, состоящей из функциональных элементов,
соответствующих этим функциям, зависит только от параметра $p$ и никак не зависит от количества используемых регистров памяти.
Более того, было показано, что любую булеву функцию можно реализовать схемой, использующей не более двух регистров памяти.

В данной работе рассматриваются схемы, состоящие из обратимых функциональных элементов NOT, CNOT и 2-CNOT.
Определение таких функциональных элементов и схем было дано, например,
в работах~\cite{feynman,maslov_thesis,my_equivalent_replacements}.
Известно, что обратимая схема с $n \geqslant 4$ входами,
состоящая из функциональных элементов NOT, CNOT и 2-CNOT (далее просто обратимая схема), задает четную подстановку на множестве
$\ZZ_2^n$ \cite{shende, my_research}.
Поэтому в качестве меры сложности четной подстановки можно рассматривать сложность задающей ее
минимальной обратимой схемы.

В данной работе рассматривается множество $F(n,q)$ всех отображений $\ZZ_2^n \to \ZZ_2^n$, которые могут быть реализованы
обратимой схемой с $(n+q)$ входами (дополнительной памяти).
Определяется функция Шеннона сложности обратимой схемы $L(n,q)$, как функция от $n$ и количества дополнительных входов схемы $q$.
Показывается, что сложность обратимой схемы, в отличие от обычных схем,
существенно зависит от количества дополнительных входов (аналог регистров памяти~\cite{karpova}).

При помощи мощностного метода Риордана-Шеннона доказывается нижняя оценка сложности обратимой схемы:
$L(n,q) \geqslant \frac{2^n(n-2)}{3\log_2(n+q)} - \frac{n}{3}$.
Дается описание алгоритма синтеза обратимой схемы без использования дополнительных входов,
при помощи которого доказывается верхняя оценка $L(n,0) \leqslant 3n2^{n+4}(1+o(1)) \mathop / \log_2n$.
Также предлагается аналог метода Лупанова~\cite{yablonsky} для синтеза обратимых схем с дополнительными входами,
при помощи которого доказывается верхняя асимптотическая оценка $L(n,q_0) \lesssim 2^n$ при $q_0 \sim n2^{n-o(n)}$.

\section{Основные понятия}

Определение обратимых функциональных элементов было введено, к примеру, в работе Фейнмана~\cite{feynman},
определения обратимых элементов NOT и $k$-CNOT были даны, к примеру, в работе~\cite{maslov_thesis}.
Mы будем пользоваться формальным определением этих функциональных элементов из работы~\cite{my_equivalent_replacements}.

Напомним, что через $N_j^n$ обозначается функциональный элемент NOT (инвертор) с $n$ входами,
задающий преобразование $\ZZ_2^n \to \ZZ_2^n$ вида
\begin{equation}
    f_j(\langle x_1, \ldots, x_n \rangle) = \langle x_1, \ldots, x_j \oplus 1, \ldots, x_n \rangle  \; .
    \label{formula_not_definition}
\end{equation}
Через $C_{i_1,\ldots,i_k;j}^n = C_{I;j}^n$, $j \notin I$, обозначается функциональный элемент $k$-CNOT с $n$ входами
(контролируемый инвертор, обобщенный элемент Тоффоли с $k$ контролирующими входами),
задающий преобразование $\ZZ_2^n \to \ZZ_2^n$ вида
\begin{equation}
    f_{i_1,\ldots,i_k;j}(\langle x_1, \ldots, x_n \rangle) =
        \langle x_1, \ldots, x_j \oplus x_{i_1} \wedge \ldots \wedge x_{i_k}, \ldots, x_n \rangle  \; .
    \label{formula_k_cnot_definition}
\end{equation}
Если значение $n$ ясно из контекста, будем опускать верхний индекс $n$ в обозначении функциональных элементов NOT и $k$-CNOT.
Далее будут рассматриваться только функциональные элементы NOT, CNOT (1-CNOT) и 2-CNOT.
Обозначим через $\Omega_n^2$ множество всех функциональных элементов NOT, CNOT и 2-CNOT с $n$ входами.

Классически схема из функциональных элементов определяется как ориентированный граф без циклов с помеченными ребрами и вершинами.
В случае обратимых схем данную модель можно упростить, т.\,к. в обратимой схеме запрещено ветвление входов и выходов
функциональных элементов, а также произвольное подключение выходов одного функционального элемента ко входам другого
функционального элемента. Поэтому в ориентированном графе, описывающем обратимую схему $\frS$, все вершины,
соответствующие функциональным элементам, имеют ровно $n$ занумерованных входов и выходов.
Все эти вершины нумеруются от 1 до $l$, при этом $i$-й выход $m$-й вершины, $m < l$,
соединяется только с $i$-м входом $(m+1)$-й вершины.
Входы 1-й вершины являются входами обратимой схемы, выходы $l$-й вершины~"--- ее выходами.
Такое соединение функциональных элементов из множества $\Omega_n^2$ друг с другом далее будем называть композицией
функциональных элементов. Величина $l=L(\frS)$ равна сложности обратимой схемы $\frS$.

Можно приписать $i$-м входам и выходам вершин графа символ $r_i$ из множества $R = \{\,r_1, \ldots, r_n\,\}$,
каждый из которых можно интерпретировать как имя регистра памяти (номер ячейки памяти),
в котором хранится часть результата работы схемы. Из формул~\eqref{formula_not_definition}
и~\eqref{formula_k_cnot_definition} видно, что в этом случае после работы какого-либо элемента схемы инвертируется
значение не более, чем в одном регистре памяти. В этом заключается существенная разница между схемами,
состоящими из обратимых и необратимых функциональных элементов.

\section{Сложность обратимой схемы}

В данном разделе будет сформулирован основной результат работы без доказательства для сложности
обратимой схемы с $n$ входами. Доказательство приведенных оценок будет дано в следующих разделах.

Обратимая схема с $n \geqslant 4$ входами задает четную подстановку на множестве $\ZZ_2^n$ \cite{shende, my_research}.
При этом она может также реализовывать некоторое булево отображение $\ZZ_2^m \to \ZZ_2^k$, где $m, k \leqslant n$,
с использованием или без использования дополнительных входов. Для пояснения этого введем следующие отображения:
\begin{enumerate}

    \item
        \textit{Расширяющее} отображение $\phi_{n,n+k}\colon \mathbb Z_2^n \to \mathbb Z_2^{n+k}$ вида
        $$
            \phi_{n,n+k}( \langle x_1, \ldots, x_n \rangle ) = \langle x_1, \ldots, x_n, 0, \ldots, 0 \rangle  \; .
        $$

    \item
        \textit{Редуцирующее} отображение $\psi_{n+k,n}^\pi\colon \mathbb Z_2^{n+k} \to \mathbb Z_2^n$ вида
        $$
            \psi_{n+k,n}^\pi( \langle x_1, \ldots, x_{n+k} \rangle ) =
            \langle x_{\pi(1)}, \ldots, x_{\pi(n)} \rangle  \; ,
        $$
        где $\pi$~"--- подстановка на множестве $\mathbb Z_{n+k}$.

\end{enumerate}

Введем формальное определение обратимой схемы, реализующей произвольное отображение $f\colon \ZZ_2^n \to \ZZ_2^n$
с использованием дополнительных входов.
\begin{definition}\label{label_define_implementation_of_transformation}
    Обратимая схема $\frS_g$ с $(n+q)$ входами, задающая преобразование $g\colon \mathbb Z_2^{n+q} \to \mathbb Z_2^{n+q}$,
    реализует отображение $f$ c использованием $q \geqslant 0$ дополнительных входов
    (дополнительной памяти), если существует такая подстановка $\pi \in S(\mathbb Z_{n+q})$, что
    $$
        \psi_{n+q,n}^\pi(g( \phi_{n,n+q}(\vv x))) = f(\vv x), \vv x \in \mathbb Z_2^n  \; .
    $$
\end{definition}
Отметим, что в данной терминологии выражения \textit{реализует} и \textit{задает} отображение имеют разные значения:
если обратимая схема $\frS_g$ задает отображение $f$, то $g(\vv x) = f(\vv x)$.
Если схема $\frS_g$ реализует отображение $f$ и имеет ровно $n$ входов, то будем говорить, что она реализует
данное отображение \textit{без использования дополнительной памяти}.

Обозначим через $P_2(n,n)$ множество всех булевых отображений $\ZZ_2^n \to \ZZ_2^n$.
Обозначим через $F(n,q) \subseteq P_2(n,n)$ множество всех отображений $\ZZ_2^n \to \ZZ_2^n$, которые могут быть реализованы
обратимой схемой с $(n+q)$ входами.
Множество подстановок из $S(\ZZ_2^n)$, задаваемых всеми элементами множества $\Omega_n^2$,
генерирует знакопеременную $A(\ZZ_2^n)$ и симметрическую $S(\ZZ_2^n)$ группы подстановок при $n > 3$ и $n \leqslant 3$,
соответственно~\cite{shende, my_research}.
Отсюда следует, что $F(n,0)$ совпадает с множеством отображений, задаваемых всеми подстановками из $A(\ZZ_2^n)$
и $S(\ZZ_2^n)$ при $n > 3$ и $n \leqslant 3$, соответственно.
С другой стороны, несложно показать, что при $q \geqslant n$ верно равенство $F(n,q) = P_2(n,n)$:
при наличии $n$ дополнительных входов всегда можно построить биекцию $g\colon \ZZ_2^{n+q} \to \ZZ_2^{n+q}$,
удовлетворяющую Определению~\ref{label_define_implementation_of_transformation}
и задающую четную подстановку $h \in A(\ZZ_2^{n+q})$.

Обозначим через $L(f,q)$ сложность минимальной обратимой схемы, состоящей из функциональных элементов
множества $\Omega_{n+q}^2$ и реализующей булево отображение $f \in F(n,q)$ с использованием $q$ дополнительных входов.
Определим функцию Шеннона $L(n,q)$ сложности обратимой схемы следующим образом:
$$
    L(n,q) = \max_{f \in F(n,q)} {L(f,q)}  \; .
$$

Теперь сформулируем основной результат данной работы.
\begin{theorem}[нижняя оценка сложности обратимой схемы]\label{theorem_complexity_lower_bound}
    Верно неравенство
    $$
        L(n,q) \geqslant \frac{2^n(n-2)}{3\log_2(n+q)} - \frac{n}{3}  \; .
    $$
\end{theorem}

\begin{theorem}[о сложности обратимой схемы без дополнительных входов]\label{theorem_complexity_upper_no_memory}
    Верно неравенство
    $$
        L(n, 0) \leqslant \frac{3n2^{n+4}}{\log_2 n - \log_2 \log_2 n - \log_2 \phi(n)}
            \left( 1 + \varepsilon(n) \right)  \; ,
    $$
    где $\phi(n) < n \mathop / \log_2 n$~"--- любая сколь угодно медленно растущая функция,
    а функция $\varepsilon(n)$ равна:
    $$
        \varepsilon(n) = \frac{1}{6\phi(n)} + \left(\frac{8}{3} + o(1)\right)
            \frac{\log_2 n \cdot \log_2 \log_2 n}{n}  \; .
    $$
\end{theorem}

\begin{theorem}\label{theorem_complexity_upper_no_memory_common}
    Верно соотношение
    $$
        L(n,0) \asymp \frac{n2^n}{\log_2n}  \; .
    $$
\end{theorem}
\begin{proof}
    Следует из Теорем~\ref{theorem_complexity_lower_bound} и~\ref{theorem_complexity_upper_no_memory}.
\end{proof}

\begin{theorem}\label{theorem_complexity_upper_with_memory}
    Верно соотношение
    $$
        L(n,q_0) \lesssim 2^n \text{ \,при\, } q_0 \sim n 2^{n-\lceil n \mathop / \phi(n)\rceil} \; ,
    $$
    где $\phi(n) \leqslant n \mathop / (\log_2 n + \log_2 \psi(n))$
    и $\psi(n)$~"--- любые сколь угодно медленно растущие функции.
\end{theorem}

\begin{theorem}\label{theorem_complexity_upper_with_memory_common}
    Верно соотношение
    $$
        L(n,q_0) \asymp 2^n \text{ \,при\, } q_0 \sim n 2^{n-\lceil n \mathop / \phi(n)\rceil} \; ,
    $$
    где $\phi(n) \leqslant n \mathop / (\log_2 n + \log_2 \psi(n))$
    и $\psi(n)$~"--- любые сколь угодно медленно растущие функции.
\end{theorem}
\begin{proof}
    Следует из Теорем~\ref{theorem_complexity_lower_bound} и~\ref{theorem_complexity_upper_with_memory}.
\end{proof}
   
Из Теорем~\ref{theorem_complexity_upper_no_memory_common} и~\ref{theorem_complexity_upper_with_memory_common}
следует важный вывод:
\begin{assertion}
    Использование дополнительной памяти в обратимых схемах, состоящих из функциональных элементов множества $\Omega_n^2$,
    почти всегда позволяет снизить сложность обратимой схемы.
\end{assertion}
\noindent Стоит отметить, что данный факт снижения сложности за счет дополнительных входов в общем случае
не был установлен для схем, состоящих из классических необратимых функциональных элементов.

\section{Нижняя оценка сложности обратимых схем}

В работах~\cite{shende,my_equivalent_replacements} было показано, что для любой подстановки $h \in A(\ZZ_2^n)$
при $n > 3$ можно построить задающую ее обратимую схему, состоящую из функциональных элементов множества $\Omega_n^2$.
Другими словами, множество подстановок, задаваемых всеми функциональными элементами из $\Omega_n^2$, $n > 3$,
генерирует знакопеременную группу $A(\ZZ_2^n)$.

В работе~\cite{gluhov} было показано, что длина $L(G,M)$ группы подстановок $G$ относительно системы образующих $M$
удовлетворяет неравенству
\begin{equation}\label{formula_gluhov_lower_bound}
    L(G,M) \geqslant \left \lceil \log_{|M|} |G| \right \rceil  \; .
\end{equation}

В нашем случае $G = A(\ZZ_2^n)$, $|G| = (2^n)!/2$, $|M| = |\Omega_n^2|$. Поскольку мощность множества $\Omega_n^2$ равна
\begin{equation}\label{formula_size_of_set_omega_n_2}
    |\Omega_n^2| = \sum_{k = 0}^2 {(n-k)  {n \choose k}} = \frac{n^3}{2} \left( 1 + o(1) \right)  \; ,
\end{equation}
то мы можем вывести простую нижнюю оценку для $L(n,0)$:
\begin{gather}
    L(n,0) \gtrsim \frac{\log_2 ((2^n)!/2)}{\log_2 (n^3/2)}
           \gtrsim \frac{\log_2 2^{n2^n} - \log_2 e^{2^n}}{3 \log_2 n}  \; , \notag \\
    L(n,0) \gtrsim \frac{n2^n}{3 \log_2 n}  \; .
    \label{formula_simple_lower_bound}
\end{gather}

Нижняя оценка~\eqref{formula_gluhov_lower_bound} в работе~\cite{gluhov} строго доказана не была и, по мнению автора,
основывается на не совсем верном предположении, что достаточно рассмотреть только все возможные произведения подстановок из $M$
длины ровно $L(G,M)$, чтобы получить все элементы группы подстановок $G$. Данное предположение верно только для системы
образующих $M$, содержащей тождественную подстановку. В противном случае, необходимо рассматривать все возможные произведения
подстановок из $M$ длины менее $L(G,M)$ в том числе. Из описания множества $\Omega_n^2$ видно, что множество подстановок,
задаваемых всеми функциональными элементами $\Omega_n^2$, не содержит тождественной подстановки.

Для того, чтобы получить общую нижнюю оценку $L(n,q)$, также необходимо учитывать те булевы отображения,
которые могут быть реализованы обратимой схемой с $(n+q)$ входами. Таких отображений не более $A_{n+q}^n$
(количество размещений из $(n+q)$ по $n$ без повторений).

Перейдем теперь непосредственно к доказательству Теоремы~\ref{theorem_complexity_lower_bound}.
\begin{proof}[Доказательство Теоремы~\ref{theorem_complexity_lower_bound}]
    $ $
    
    Докажем при помощи мощностного метода Риордана-Шеннона, что верно неравенство
    $$
        L(n,q) \geqslant \frac{2^n(n-2)}{3\log_2(n+q)} - \frac{n}{3}  \; .
    $$
    
    Пусть $r = |\Omega_n^2|$. Из формулы~\eqref{formula_size_of_set_omega_n_2} следует, что
    \begin{gather*}
        r = \sum_{k=0}^2{(n-k)\binom{n}{k}} = \frac{n^3 - n^2 + 2n}{2}  \; , \notag \\
        \frac{n^2(n-1)}{2} + 1 < r \leqslant \frac{n^3}{2} \text{ \,при\, } n \geqslant 2  \; .
    \end{gather*}

    Обозначим через $\EuScript C^*(n,s) = r^s$ и $\EuScript C(n,s)$ количество всех обратимых схем,
    состоящих из функциональных элементов множества $\Omega_n^2$, сложность которых равна $s$ и не превышает $s$, соответственно.
    Тогда
    \begin{align*}
        \EuScript C(n,s) &= \sum_{i=0}^s {\EuScript C^*(n,i)} = \frac{r^{s+1} - 1}{r-1}
            \leqslant \left( \frac{n^3}{2} \right)^{s+1} \cdot \frac{2}{n^2(n-1)}  \; , \\
        \EuScript C(n,s) &\leqslant \left( \frac{n^3}{2} \right)^s \cdot \left(1 + \frac{1}{n-1}\right)
            \text{ \,при\, } n \geqslant 2  \; .
    \end{align*}

    Как было сказано выше, каждой обратимой схеме с $(n+q)$ входами
    соответствует не более $A_{n+q}^n$ различных булевых отображений $\ZZ_2^n \to \ZZ_2^n$.
    Следовательно, верно следующее неравенство:
    $$
       \EuScript C(n+q,L(n,q)) \cdot A_{n+q}^n \geqslant |F(n,q)|  \; .
    $$
    Поскольку $|F(n,q)| \geqslant |A(\ZZ_2^n)| = (2^n)! \mathop / 2$ и $A_{n+q}^n \leqslant (n+q)^n$, то
    $$
        \left( \frac{(n+q)^3}{2} \right)^{L(n,q)} \cdot \left(1 + \frac{1}{n+q-1}\right)
            \cdot (n+q)^n \geqslant (2^n)! \mathop / 2  \; .
    $$
    
    Несложно убедиться, что при $n > 0$ верно неравенство $(2^n)! \geqslant (2^n \mathop / e)^{2^n}$.
    Следовательно,
    \begin{multline*}
        L(n,q) \cdot (3\log_2(n+q) - 1) + \log_2 \left(1 + \frac{1}{n+q-1}\right) + \\
            + n \log_2(n+q) \geqslant 2^n(n - \log_2 e)  \; .
    \end{multline*}
    
    Отсюда следует неравенство из условия теоремы
    $$
        L(n,q) \geqslant \frac{2^n(n-2)}{3\log_2(n+q)} - \frac{n}{3}  \; .
    $$
\end{proof}


\section{Верхняя оценка сложности обратимых схем без дополнительных входов}

В работе~\cite{my_fast_synthesis_algorithm} был предложен алгоритм синтеза обратимой схемы,
состоящей из функциональных элементов множества $\Omega_n^2$ и задающей подстановку $h \in A(\ZZ_2^n)$,
использующий теорию групп подстановок. Данный алгоритм синтеза основан на представлении подстановки $h$
в виде произведения пар независимых транспозиций. Было показано, что схема $\frS$, синтезированная данным алгоритмом,
имеет сложность $L(\frS) \lesssim 7n 2^n$. Отсюда можно вывести простую верхнюю оценку для $L(n,0)$:
\begin{equation}\label{formula_simple_upper_bound_no_memory}
    L(n,0) \lesssim 7n 2^n  \; .
\end{equation}
Если взять за основу данный подход синтеза, то верхнюю оценку~\eqref{formula_simple_upper_bound_no_memory}
можно существенно улучшить.

\begin{proof}[Доказательство Теоремы~\ref{theorem_complexity_upper_no_memory}]
    Доказательство основано на описании алгоритма синтеза, позволяющего получить для любой четной подстановки $h \in A(\ZZ_2^n)$
    задающую ее обратимую схему $\frS$ со сложностью:
    $$
        L(\frS) \leqslant \frac{3n2^{n+4}}{\log_2 n - \log_2 \log_2 n - \log_2 \phi(n)}
            \left( 1 + \varepsilon(n) \right)  \; ,
    $$
    где $\phi(n) < n \mathop / \log_2 n$~"--- любая сколь угодно медленно растущая функция,
    а функция $\varepsilon(n)$ равна:
    $$
        \varepsilon(n) = \frac{1}{6\phi(n)} +\left(\frac{8}{3} + o(1)\right)
            \frac{\log_2 n \cdot \log_2 \log_2 n}{n}  \; .
    $$

    Каждую подстановку $h \in A(\ZZ_2^n)$ можно представить в виде произведения независимых циклов,
    причем сумма длин этих циклов не превосходит $2^n$.
    Произведение двух независимых циклов можно выразить следующим образом:
    \begin{multline}
        \label{formula_decompostion_of_two_cycles}
        (i_1, i_2, \ldots, i_{l_1}) \circ (j_1, j_2, \ldots, j_{l_2}) = \\
         = (i_1, i_2) \circ (j_1, j_2) \circ (i_1, i_3, \ldots, i_{l_1}) \circ (j_1, j_3, \ldots, j_{l_2})  \; .
    \end{multline}
    Цикл длины $l \geqslant 5$ можно выразить следующим образом:
    \begin{equation}
        \label{formula_decompostion_of_k_cycle}
        (i_1, i_2, \ldots, i_l) = (i_1, i_2) \circ (i_3, i_4) \circ (i_1, i_3, i_5, i_6, \ldots, i_l)  \; .
    \end{equation}

    Представим подстановку $h$ в виде произведения независимых транспозиций, разбитых на группы по $K$ транспозиций в каждой,
    и некоторой остаточной подстановки $h'$:
    \begin{equation}
        h = \mcirc_{\vv x_t, \vv y_t \in \ZZ_2^n} {\left( (\vv x_1, \vv y_1) \circ \ldots
        \circ (\vv x_K, \vv y_K) \right )} \circ h'  \; .
        \label{formula_permutation_decomposition_main}
    \end{equation}
    Оценим количество независимых циклов и их длину в представлении подстановки $h'$.
    Согласно формулам~\eqref{formula_decompostion_of_two_cycles} и~\eqref{formula_decompostion_of_k_cycle} из подстановки $h'$
    нельзя получить $K$ независимых транспозиций, если количество независимых циклов в ее представлении строго меньше $K$
    и их длина строго меньше 5-ти.
    Таким образом, сумма длин циклов в представлении $h'$ не превосходит $4(K-1)$.

    Обозначим через $M_g$ множество подвижных точек подстановки $g \in S(\ZZ_2^n)$:
    $$
        M_g = \{\,\vv x \in \ZZ_2^n\mid g(\vv x) \ne \vv x\,\}  \; .
    $$
    Тогда $|M_h| \leqslant 2^n$, $|M_{h'}| \leqslant 4(K-1)$.

    Из формул~\eqref{formula_decompostion_of_two_cycles}--\eqref{formula_permutation_decomposition_main}
    следует, что в представлении подстановки $h$ в виде произведения траспозиций можно получить не более
    $|M_h| \mathop / K$ групп,
    в каждой из которых $K$ независимых транспозиций, а в представлении подстановки $h'$ в виде произведения траспозиций
    можно получить не более $|M_{h'}| \mathop / 2$ пар независимых транспозиций и не более одной пары зависимых транспозиций.
    Пара зависимых транспозиций $(i,j) \circ (i, k)$ выражается через произведение двух пар независимых транспозиций:
    $$
        (i,j) \circ (i, k) = ((i,j) \circ (r, s)) \circ ((r,s) \circ (i,k))  \; .
    $$

    Обозначим через $f_h$ булево отображение $\ZZ_2^n \to \ZZ_2^n$, соответствующее подстановке $h$.
    Тогда можно оценить сверху $L(f_h,0)$ следующим образом:
    \begin{gather}
        L(f_h,0) \leqslant \frac{|M_h|}{K} \cdot L(f_{g^{(K)}},0) + \left(\frac{|M_{h'}|}{2} + 2 \right )
            \cdot L(f_{g^{(2)}},0)   \; , \notag \\
        L(f_h,0) \leqslant \frac{2^n}{K} L(f_{g^{(K)}},0) + 2K \cdot L(f_{g^{(2)}},0)  \; .
            \label{formula_upper_bound_of_L_h_common}
    \end{gather}
    где $g^{(i)}$~"--- произвольная подстановка, представляющая собой произведение $i$ независимых транспозиций.
    Опишем алгоритм синтеза, позволяющий получить обратимую схему $\frS$, реализующую отображение $f_h$.

    Рассмотрим произвольную подстановку $g^{(K)}$.
    Обозначим через $k$ величину $|M_{g^{(K)}}|$, тогда $k = 2K$.
    Суть описываемого алгоритма заключается в действии сопряжением на подстановку $g^{(K)}$
    таким образом, чтобы получить некоторую новую подстановку, соответствующую одному обобщенному элементу Тоффоли.
    Напомним, что действие сопряжением не меняет цикловой структуры подстановки, поэтому подстановка $g^{(K)}$ в результате
    действия сопряжением всегда будет оставаться произведением $K$ независимых транспозиций.
    Любой элемент $E$ из множества $\Omega_n^2$ задает подстановку $h_{E}$ на множестве двоичных векторов $\ZZ_2^n$.
    Для этой подстановки верно равенство $h^{-1}_{E} = h_{E}$.
    Следовательно, применение к $g^{(K)}$ действия сопряжением подстановкой $h_{E}$,
    записываемое как $h^{-1}_{E} \circ g^{(K)} \circ h_{E}$, соответствует
    присоединению элемента $E$ к началу и к концу текущей обратимой подсхемы.

    Пусть $g^{(K)} = (\vv x_1, \vv y_1) \circ \ldots \circ (\vv x_K, \vv y_K)$.
    Составим матрицу $A$ следующим образом:
    \begin{equation}
        A =
            \left(
                \begin{matrix}
                    \vv x_1 \\
                    \vv y_1 \\
                    \ldots \\
                    \vv x_K \\
                    \vv y_K
                \end{matrix}
            \right )
          =
            \left(
                \begin{matrix}
                    a_{1,1}   & \ldots & a_{1,n}   \\
                    a_{2,1}   & \ldots & a_{2,n}   \\
                    \hdotsfor{3}                   \\
                    a_{k-1,1} & \ldots & a_{k-1,n} \\
                    a_{k,1}   & \ldots & a_{k,n}   \\
                \end{matrix}
            \right )  \; .
        \label{formula_matrix_for_permutation}
    \end{equation}

    Наложим на значение $k$ следующее ограничение: $k$ должно быть степенью двойки, $2^{\lfloor \log_2 k \rfloor} = k$.
    Если $k \leqslant \log_2 n$, то в матрице $A$ существует не более $2^k$ и не менее $\log_2 k$ попарно различных столбцов.
    Без ограничения общности будем считать, что такими столбцами являются первые $d \leqslant 2^k$ столбцов матрицы.
    Тогда для любого $j$-го столбца, $j > d$, найдется равный ему $i$-й столбец, $i \leqslant d$.
    Следовательно, применив к подстановке $g^{(K)}$ действие сопряжением подстановкой,
    задаваемой функциональным элементом $C_{i;j}$,
    можно обнулить $j$-й столбец в матрице $A$ (для этого потребуется 2 элемента CNOT).
    Обнуляя таким образом все столбцы с индексами больше $d$, использовав $L_1 \leqslant 2(n-d)$ функциональных элементов CNOT,
    мы получим новую подстановку $g_1^{(K)}$ и соответствующую ей матрицу $A_1$ следующего вида:
    $$
        A_1 =
            \left(\phantom{\begin{matrix} 0\\0\\0\\0\\0\\ \end{matrix}} \right.
            \begin{matrix}
                a_{1,1}   & \ldots & a_{1,d}   \\
                a_{2,1}   & \ldots & a_{2,d}   \\
                \hdotsfor{3}                     \\
                a_{k-1,1} & \ldots & a_{k-1,d} \\
                a_{k,1}   & \ldots & a_{k,d}   \\
            \end{matrix}
            \phantom{\begin{matrix} 0\\0\\0\\0\\0\\ \end{matrix}}
            \overbrace{
                \begin{matrix}
                    0 &\ldots & 0   \\
                    0 &\ldots & 0   \\
                    \hdotsfor{3}    \\
                    0 &\ldots & 0   \\
                    0 &\ldots & 0   \\
                \end{matrix}
            }^{n - d}
            \left.\phantom{\begin{matrix} 0\\0\\0\\0\\0\\ \end{matrix}} \right)  \; .
    $$

    Теперь для всех $a_{1,i} = 1$ применяем к $g_1^{(K)}$ действие сопряжением подстановкой,
    задаваемой функциональным элементом $N_i$. Для этого потребуется $L_2 \leqslant 2d$ элементов NOT.
    В итоге получим подстановку $g_2^{(K)}$ и соответствующую ей матрицу
    $A_2$ (элементы матрицы обозначены через $b_{i,j}$, чтобы показать их возможное отличие от элементов матрицы $A_1$):
    $$
        A_2 =
            \left(\phantom{\begin{matrix} 0\\0\\0\\0\\0\\ \end{matrix}} \right.
            \begin{matrix}
                0         & \ldots & 0         \\
                b_{2,1}   & \ldots & b_{2,d}   \\
                \hdotsfor{3}                   \\
                b_{k-1,1} & \ldots & b_{k-1,d} \\
                b_{k,1}   & \ldots & b_{k,d}   \\
            \end{matrix}
            \phantom{\begin{matrix} 0\\0\\0\\0\\0\\ \end{matrix}}
            \overbrace{
                \begin{matrix}
                    0 &\ldots & 0   \\
                    0 &\ldots & 0   \\
                    \hdotsfor{3}    \\
                    0 &\ldots & 0   \\
                    0 &\ldots & 0   \\
                \end{matrix}
            }^{n - d}
            \left.\phantom{\begin{matrix} 0\\0\\0\\0\\0\\ \end{matrix}} \right)  \; .
    $$

    Следующим шагом является приведение матрицы $A_2$ к \textit{каноническому виду}, где каждая строка, если ее записать
    в обратном порядке, представляет собой запись в двоичной системе счисления числа <<номер строки минус 1>>.

    Все строки матрицы $A_2$ различны. Первая строка уже имеет канонический вид, поэтому мы последовательно будем приводить
    оставшиеся строки к каноническому виду, начиная со второй.
    Предположим, что текущая строка имеет номер $i$, и все строки с номерами от $1$ до $(i-1)$ имеют канонический вид.
    Возможны два случая:
    \begin{enumerate}
        \item
            Существует ненулевой элемент в $i$-й строке с индексом $j > \log_2 k$: $b_{i,j} = 1$.
            В этом случае для всех элементов матрицы $b_{i, j'}$, $j' \ne j$, $j' \leqslant d$,
            не равных $j'$-ой цифре в двоичной записи числа $(i-1)$, мы применяем к $g_2^{(K)}$ действие сопряжением
            подстановкой, задаваемой функциональным элементом $C_{j;j'}$. Для этого потребуется не более $2d$ элементов CNOT.
            После этого нам остается только обнулить $j$-й элемент текущей строки.
            Для этого мы применяем к $g_2^{(K)}$ действие сопряжением подстановкой, задаваемой
            функциональным элементом $C_{I;j}$, где $I$~"--- множество индексов ненулевых цифр в двоичной записи числа $(i-1)$.
            К примеру, если $i = 6$, то $I = \{\,1, 3\,\}$.
            Поскольку $|I| \leqslant \log_2 k$, мы можем заменить данный функциональный элемент $C_{I;j}$
            композицией не более $8 \log_2 k$ функциональных элементов 2-CNOT~\cite{shende}.
            Следовательно, для данного действия сопряжением нам потребуется не более $16 \log_2 k$ элементов 2-CNOT.
            
            Итак, суммируя количество используемых функциональных элементов, мы получаем, что для приведения $i$-й строки
            к каноническому виду в данном случае требуется $L_3^{(i)} \leqslant 2d + 16 \log_2 k$
            элементов из множества $\Omega_n^2$.

        \item
            Не существует ненулевого элемента в $i$-й строке с индексом $j > \log_2 k$: $b_{i,j} = 0$ для всех $j > \log_2 k$.
            В этом случае мы применяем к $g_2^{(K)}$ действие сопряжением подстановкой,
            задаваемой функциональным элементом $C_{I;\log_2 k + 1}$, где $I$~--- множество индексов ненулевых элементов
            текущей строки.
            Т.\,к. все строки матрицы различны и при этом все предыдущие строки находятся в каноническом виде,
            мы можем утверждать, что значение элемента матрицы $b_{j,\log_2 k + 1}$ после данного действия сопряжением
            будет изменено только в случае, если $j \geqslant i$.
            Поскольку $|I| \leqslant \log_2 k$, мы можем заменить данный функциональный элемент $C_{I;j}$
            композицией не более $8 \log_2 k$ функциональных элементов 2-CNOT~\cite{shende}.
            Следовательно, для данного действия сопряжением нам потребуется не более $16 \log_2 k$ элементов 2-CNOT.
            После этого мы можем перейти к предыдущему случаю.

            Итак, суммируя количество используемых функциональных элементов, мы получаем, что для приведения $i$-й строки
            к каноническому виду в данном случае требуется $L_3^{(i)} \leqslant 2d + 32 \log_2 k$
            элементов из множества $\Omega_n^2$.
    \end{enumerate}

    После приведения матрицы $A_2$ к каноническому виду, мы получим новую подстановку $g_3^{(K)}$
    и соответствующую ей матрицу $A_3$ следующего вида:
    $$
        A_3 =
            \left(\phantom{\begin{matrix} 0\\0\\0\\0\\0\\ \end{matrix}} \right.
            \overbrace{
                \begin{matrix}
                    0 & 0 & 0 & \ldots & 0 \\
                    1 & 0 & 0 & \ldots & 0 \\
                    \hdotsfor{5}           \\
                    0 & 1 & 1 & \ldots & 1 \\
                    1 & 1 & 1 & \ldots & 1 \\
                \end{matrix}
            }^{\log_2 k}        
            \phantom{\begin{matrix} 0\\0\\0\\0\\0\\ \end{matrix}}
            \overbrace{
                \begin{matrix}
                    0 &\ldots & 0   \\
                    0 &\ldots & 0   \\
                    \hdotsfor{3}    \\
                    0 &\ldots & 0   \\
                    0 &\ldots & 0   \\
                \end{matrix}
            }^{n - \log_2 k}
            \left.\phantom{\begin{matrix} 0\\0\\0\\0\\0\\ \end{matrix}} \right)  \; .
    $$
    Для этого в сумме потребуется $L_3$ функциональных элементов множества $\Omega_n^2$:
    $$
        L_3 = \sum_{i=2}^k {L_3^{(i)}} \leqslant k(2d + 32 \log_2 k)  \; .
    $$
    При этом мы получили еще одно ограничение на значение $k$: значение $\log_2 k$ должно быть строго меньше $n$,
    иначе не всегда будет возможно привести матрицу $A_2$ к каноническому виду.

    На последнем шаге для каждого $i > \log_2 k$ мы применяем к $g_3^{(K)}$ действие сопряжением подстановкой,
    задаваемой функциональным элементом $N_i$. Для этого нам потребуется $L_4 = 2(n - \log_2 k)$ элементов NOT.
    В итоге получим подстановку $g_4^{(K)}$ и соответствующую ей матрицу $A_4$ следующего вида:
    $$
        A_4 =
            \left(\phantom{\begin{matrix} 0\\0\\0\\0\\0\\ \end{matrix}} \right.
            \overbrace{
                \begin{matrix}
                    0 & 0 & 0 & \ldots & 0 \\
                    1 & 0 & 0 & \ldots & 0 \\
                    \hdotsfor{5}           \\
                    0 & 1 & 1 & \ldots & 1 \\
                    1 & 1 & 1 & \ldots & 1 \\
                \end{matrix}
            }^{\log_2 k}        
            \phantom{\begin{matrix} 0\\0\\0\\0\\0\\ \end{matrix}}
            \overbrace{
                \begin{matrix}
                    1 &\ldots & 1   \\
                    1 &\ldots & 1   \\
                    \hdotsfor{3}    \\
                    1 &\ldots & 1   \\
                    1 &\ldots & 1   \\
                \end{matrix}
            }^{n - \log_2 k}        
            \left.\phantom{\begin{matrix} 0\\0\\0\\0\\0\\ \end{matrix}} \right)  \; .
    $$

    Подстановка $g_4^{(K)}$ задается одним функциональным элементом $C_{n,n-1, \ldots, \log_2 k + 1; 1}$.
    Этот элемент имеет $(n - \log_2 k)$ контролирующих входов, поэтому он может быть заменен композицией не более
    $L_5 \leqslant 8(n - \log_2 k)$ функциональных элементов 2-CNOT~\cite{shende}.
    
    Мы получили подстановку $g_4^{(K)}$, применяя к $g^{(K)}$ действие сопряжением подстановками определенного вида.
    Если мы применим к $g_4^{(K)}$ действие сопряжением в точности теми же подстановками, но в обратном порядке,
    мы получим $g^{(K)}$. В терминах синтеза обратимой логики это означает, что мы должны присоединить
    ко входу и выходу функционального элемента $C_{n,n-1, \ldots, \log_2 k + 1; 1}$ все те функциональные элементы,
    что мы использовали в наших преобразованиях исходной матрицы $A$, но в обратном порядке,
    и как результат, мы получим обратимую схему $\frS_K$, задающую подстановку $g^{(K)}$.
    
    Таким образом, можно утверждать, что $L(g^{(K)},0) \leqslant L(\frS_K)$ и
    \begin{multline*}
        L(g^{(K)},0) \leqslant \sum_{i=1}^5 {L_i} \leqslant 2(n-d) + 2d + \\
            + k(2d + 32 \log_2 k) + 2(n-\log_2 k) + 8(n - \log_2 k)  \; ,
    \end{multline*}
    
    \begin{equation}\label{formula_upper_bound_in_synthesis_algorithm}
        L(g^{(K)},0) \leqslant 12n + k2^{k+1} + 32k\log_2 k - 10 \log_2 k  \; .
    \end{equation}
    Отсюда также следует, что $L(g^{(2)},0) \leqslant 12n + 364$.

    Подставляя полученные верхние оценки в формулу~\eqref{formula_upper_bound_of_L_h_common},
    мы получаем следующую верхнюю оценку для $L(f_h,0)$:
    \begin{equation}\label{formula_L_h_upper_bound_general_case_with_k}
        L(f_h,0) \leqslant \frac{2^{n+1}}{k}(12n + k2^{k+1} + 32k\log_2 k - 10 \log_2 k) + k(12n + 364)  \; .
    \end{equation}

    Описанным алгоритмом требуется, чтобы $k$ было степенью двойки и чтобы $\log_2 k$ было строго меньше $n$.
    Пусть $m = \log_2 n - \log_2 \log_2 n - \log_2 \phi(n)$ и $k = 2^{\lfloor \log_2 m \rfloor}$,
    где $\phi(n) < n \mathop / \log_2 n$~"--- сколь угодно медленно растущая функция.
    Тогда $m / 2 \leqslant k \leqslant m$ и
    \begin{gather*}
        L(f_h,0) \leqslant \frac{2^{n+2}}{m}(12n + 2m2^m + 32m\log_2 m) + m(12n + 364)  \; , \\
        L(f_h,0) \leqslant \frac{3n2^{n+4}}{m} 
            \left( 1 + \frac{2^m\log_2 n}{6n} + \left(\frac{8}{3} + o(1)\right)
            \frac{\log_2 n \cdot \log_2 \log_2 n}{n} \right)  \; .
    \end{gather*}

    Отсюда следует итоговая верхняя оценка для $L(f_h,0)$:
    $$
        L(f_h,0) \leqslant \frac{3n2^{n+4}}{\log_2 n - \log_2 \log_2 n - \log_2 \phi(n)}
            \left( 1 + \varepsilon(n) \right)  \; ,
    $$
    где функция $\varepsilon(n)$ равна:
    $$
        \varepsilon(n) = \frac{1}{6\phi(n)} +\left(\frac{8}{3} + o(1)\right)
            \frac{\log_2 n \cdot \log_2 \log_2 n}{n}  \; .
    $$

    Поскольку мы описали алгоритм синтеза обратимой схемы $\frS$ для произвольной подстановки $h$,
    то $L(n,0) \leqslant L(f_h,0)$.
\end{proof}

Отметим также, что если представлять подстановку $h \in A(\ZZ_2^n)$ в виде произведения пар независимых транспозиций,
то в этом случае задающая ее обратимая схема $\frS$, синтезируемая описанным алгоритмом,
согласно формуле~\eqref{formula_L_h_upper_bound_general_case_with_k} будет иметь сложность $L(\frS) \lesssim 6n 2^n$.
Данная сложность асимптотически ниже, чем сложность обратимой схемы, синтезированной алгоритмом
из работы~\cite{my_fast_synthesis_algorithm} (см. формулу~\eqref{formula_simple_upper_bound_no_memory}).


\section{Верхняя оценка сложности обратимых схем с дополнительными входами}
    
Функциональный элемент $k$-CNOT при $k < (n-1)$ можно заменить композицией не более $8k$ элементов 2-CNOT~\cite{shende},
если не использовать дополнительные входы. Однако если использовать $(k-2)$ дополнительных входов, то элемент $k$-CNOT
при любом значении $k < n$ можно заменить композицией $(2k-3)$ элементов 2-CNOT.
При этом после такой замены на всех дополнительных выходах будет значение 0, поэтому их можно будет использовать в дальнейшем.
Если же элемент $k$-CNOT заменить композицией $(k-1)$ элементов 2-CNOT с использованием $(k-2)$
дополнительных входов, то на дополнительных выходах после замены могут быть значения, отличные от 0. Как следствие,
эти дополнительные выходы нельзя будет использовать в дальнейшем.

Таким образом, если в алгоритме синтеза, описанном в предыдущем разделе, использовать ровно $(n-3)$ дополнительных входов,
то в формуле~\eqref{formula_upper_bound_in_synthesis_algorithm} слагаемое $12n = 4n + 8n$ можно заменить на $6n = 4n + 2n$.
В этом случае из формулы~\eqref{formula_L_h_upper_bound_general_case_with_k} следует, что
$L(n,n-3) \leqslant 3n2^{n+3}(1+o(1)) \mathop / \log_2n$.
Если же в описанном алгоритме синтеза использовать $q_0 \geqslant (n-3) 2^{n+2} / (\log_2 n - \log_2 \log_2 n - \log_2 \phi(n))$
дополнительных входов, где $\phi(n) < n \mathop / \log_2 n$~"--- сколь угодно медленно растущая функция,
то в формуле~\eqref{formula_upper_bound_in_synthesis_algorithm} слагаемое $12n = 4n + 8n$ можно заменить на $5n = 4n + n$.
В этом случае из формулы~\eqref{formula_L_h_upper_bound_general_case_with_k} следует, что
$L(n,q_0) \leqslant 5n 2^{n+2} / \log_2 n$. Однако можно получить существенно меньшую верхнюю оценку для $L(n,q)$ при
использовании гораздо меньшего количества дополнительных входов, что и будет показано далее.

Лупановым О.\,Б. был предложен асимптотически наилучший метод синтеза схемы из функциональных элементов
в базисе $\{\,\neg, \wedge, \vee\,\}$, реализующей заданную булеву функцию~\cite{yablonsky}.
Было доказано, что для булевой функции от $n$ переменных сложность схемы эквивалентна $2^n / n$.
Воспользуемся данным результатом и применим аналогичный подход для синтеза обратимой схемы, состоящей из функциональных элементов
множества $\Omega_n^2$ и реализующей булево отображение $f \in F(n,q)$ с использованием $q$ дополнительных входов.

Базис функциональных элементов $\{\,\neg, \oplus, \wedge\,\}$ является полным. Каждый элемент этого базиса
можно выразить через композицию функциональных элементов NOT, CNOT и 2-CNOT. Из рис.~\ref{pic_basis} видно,
что для этого требуется не более двух функциональных элементов и не более одного дополнительного входа.
\begin{figure}[ht]
    \centering
    \includegraphics[scale=1.2]{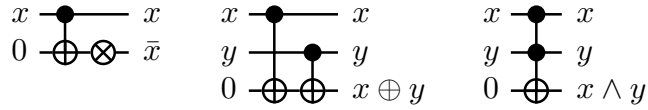}
    \caption
    {
        Выражение функциональных элементов базиса $\{\,\neg, \oplus, \wedge\,\}$ через
        композицию функциональных элементов NOT, CNOT и 2-CNOT.
    }\label{pic_basis}
\end{figure}

Также нам потребуется следующая лемма о сложности обратимой схемы, реализующей все конъюнкции $n$ переменных вида
$x_1^{a_1} \wedge \ldots \wedge x_n^{a_n}$, $a_i \in \ZZ_2$.
\begin{lemma}\label{lemma_complexity_of_all_conjunctions_of_n_variables}
    Все конъюнкции $n$ переменных вида $x_1^{a_1} \wedge \ldots \wedge x_n^{a_n}$, $a_i \in \ZZ_2$,
    можно реализовать обратимой схемой $\frS_n$, состоящей из функциональных элементов множества $\Omega_{n+q}^2$
    и имеющей сложность $L(\frS_n) \sim 2^n$ при использовании $q(\frS_n) \sim 2^n$ дополнительных входов.
\end{lemma}
\begin{proof}
    Сперва мы реализуем все инверсии $\bar x_i$, $1 \leqslant i \leqslant n$.
    Это может быть сделано при помощи $L_1 = 2n$ элементов NOT и CNOT при использовании $q_1 = n$ дополнительных входов.

    Искомую обратимую схему $\frS_n$ мы строим следующим образом:
    при помощи обратимых схем $\frS_{\lceil n \mathop / 2 \rceil}$ и $\frS_{\lfloor n \mathop / 2 \rfloor}$
    мы реализуем все конъюнкции первых $\lceil n \mathop / 2 \rceil$ и последних $\lfloor n \mathop / 2 \rfloor$
    переменных. Затем мы реализуем конъюнкции выходов этих двух схем каждого с каждым.
    Для этого потребуется $L_2 = 2^n$ элементов 2-CNOT и $q_2 = 2^n$ дополнительных входов.

    Отсюда следует, что
    $$
        L(\frS_n) = q(\frS_n) = 2^n + L(\frS_{\lceil n \mathop / 2 \rceil}) + L(\frS_{\lfloor n \mathop / 2 \rfloor})
            = 2^n(1 + o(1))  \; .
    $$
\end{proof}

Перейдем теперь непосредственно к доказательству Теоремы~\ref{theorem_complexity_upper_with_memory}.
\begin{proof}[Доказательство Теоремы~\ref{theorem_complexity_upper_with_memory}]
    Опишем алгоритм синтеза обратимой схемы $\frS$, реализующей заданное булево отображение $f\colon \ZZ_2^n \to \ZZ_2^n$
    со сложностью $L(\frS) \lesssim 2^n$ при использовании $q_0 \sim n 2^{n-\lceil n \mathop / \phi(n)\rceil}$
    дополнительных входов, где $\phi(n) \leqslant n \mathop / (\log_2 n + \log_2 \psi(n))$
    и $\psi(n)$~"--- любые сколь угодно медленно растущие функции.
    
    Отображение $f$ можно представить следующим образом:
    \begin{equation}\label{formula_function_decomposition_by_last_variables}
        f(\vv x) = \bigoplus_{a_{k+1}, \ldots, a_n \in \ZZ_2} {x_{k+1}^{a_{k+1}} \wedge \ldots \wedge x_n^{a_n}}
        \wedge f(\langle x_1, \ldots, x_k, a_{k+1}, \ldots, a_n \rangle)  \; .
    \end{equation}
    Каждое из $2^{n-k}$ булевых отображений
    $f_i(\langle x_1, \ldots, x_k \rangle) = f(\langle x_1, \ldots, x_k, a_{k+1}, \ldots, a_n \rangle)$,
    где $\sum_{j=1}^{n-k} {a_{k+j} 2^{j-1}} = i$,
    является отображением вида $\ZZ_2^k \to \ZZ_2^n$ и может быть представлено системой $n$ координатных функций
    $f_{i,j}(\vv x)\colon \ZZ_2^k \to \ZZ_2$, $\vv x \in \ZZ_2^k$, $1 \leqslant j \leqslant n$.

    Каждая координатная функция $f_{i,j}(\vv x)$ может быть получена при помощи аналога СДНФ, в котором дизъюнкции заменяются
    на сложение по модулю 2:
    \begin{equation}\label{formula_analog_sdnf}
        f_{i,j}(\vv x) = \bigoplus_{
            \substack{\boldsymbol \sigma \in \ZZ_2^k \\f_{i,j}(\boldsymbol \sigma) = 1}}
            x_1^{\sigma_1} \wedge \ldots \wedge x_k^{\sigma_k}  \; .
    \end{equation}

    Все $2^k$ конъюнкций вида $x_1^{\sigma_1} \wedge \ldots \wedge x_k^{\sigma_k}$ можно разделить на группы,
    в каждой из которых будет не более $s$ конъюнкций. Обозначим через $p$ количество таких групп:
    $p = \lceil 2^k \mathop / s \rceil$.
    Используя конъюнкции одной группы, мы можем реализовать не более $2^s$ булевых функций по формуле%
    ~\eqref{formula_analog_sdnf}.
    Обозначим через $G_i$ множество булевых функций, которые могут быть реализованы при помощи конъюнкций $i$-й группы,
    $1 \leqslant i \leqslant p$. Тогда $|G_i| \leqslant 2^s$.
    Следовательно, мы можем переписать формулу~\eqref{formula_analog_sdnf} следующим образом:
    \begin{equation}
        f_{i,j}(\vv x) = \bigoplus_{
            \substack{t=1 \ldots p\\ g_{j_t} \in G_t\\ 1 \leqslant j_t \leqslant |G_t|}} g_{j_t}(\vv x)  \; .
        \label{formula_analog_sdnf_improved}
    \end{equation}

    Отметим, что все булевы функции множества $G_i$ можно реализовать, используя такой же подход, что и в Лемме%
    ~\ref{lemma_complexity_of_all_conjunctions_of_n_variables}. В этом случае каждый элемент 2-CNOT просто заменяется
    композицией двух элементов CNOT.
    Суммарно нам потребуется $L \lesssim 2^{s+1}$ элементов CNOT и $q \sim 2^s$ дополнительных входов.

    Описываемый алгоритм синтеза конструирует обратимую схему $\frS$, реализующую
    булево отображение $f$~\eqref{formula_function_decomposition_by_last_variables}, при помощи следующих подсхем:
    \begin{enumerate}
        \item\label{item_first_subcircuit_min_complexity}
            Подсхема $\frS_1$, реализующая все конъюнкции первых $k$ переменных $x_i$ по Лемме%
            ~\ref{lemma_complexity_of_all_conjunctions_of_n_variables}
            со сложностью $L_1 \sim 2^k$ при использовании $q_1 \sim 2^k$ дополнительных входов.

        \item
            Подсхема $\frS_2$, реализующая все конъюнкции последних $(n-k)$ переменных $x_i$ по Лемме%
            ~\ref{lemma_complexity_of_all_conjunctions_of_n_variables}
            со сложностью $L_2 \sim 2^{n-k}$ при использовании $q_2 \sim 2^{n-k}$ дополнительных входов.

        \item
            Подсхема $\frS_3$, реализующая все булевы функции $g \in G_i$ для всех $i \in \ZZ_p$
            по формуле~\eqref{formula_analog_sdnf} со сложностью $L_3 \sim p2^{s+1}$ при использовании
            $q_3 \sim p2^s$ дополнительных входов (см. замечание выше про реализацию всех булевых функций множества $G_i$).
        
        \item
            Подсхема $\frS_4$, реализующая все $n2^{n-k}$ координатных функций $f_{i,j}(\vv x)$,
            $i \in \ZZ_{2^{n-k}}$, $j \in \ZZ_n$, по формуле~\eqref{formula_analog_sdnf_improved}
            со сложностью $L_4 \leqslant pn 2^{n-k}$ при использовании $q_4 = n 2^{n-k}$ дополнительных входов.
            
        \item
            Подсхема $\frS_5$, реализующая булево отображение $f$
            по формуле~\eqref{formula_function_decomposition_by_last_variables}
            со сложностью $L_5 \leqslant n 2^{n-k}$ при использовании $q_5 = n$ дополнительных вдохов.
    \end{enumerate}
    
    Будем искать параметры $k$ и $s$, удовлетворяющие следующим условиям:
    $$
        \left\{
            \begin{array}{lr}
                s = n - 2k  \; , & \\
                k = \lceil n \mathop / \phi(n) \rceil \;, & \text{где $\phi(n)$~--- некоторая растущая функция,} \\
                1 \leqslant s < n   \; , & \\
                1 \leqslant k < n \mathop / 2   \; , & \\
                \frac{2^k}{s} \geqslant \psi(n) \;, & \text{где $\psi(n)$~--- некоторая растущая функция.}
            \end{array}
        \right.
    $$
    В этом случае $p = \lceil 2^k \mathop / s \rceil \sim 2^k \mathop / s$
    и $2^{\lceil n \mathop / \phi(n) \rceil} \geqslant s\psi(n)$,
    откуда следует, что при $\phi(n) \leqslant n \mathop / (\log_2n + \log_2 \psi(n))$ параметры $k$ и $s$ будут удовлетворять
    условиям выше.
    
    Суммируя сложности обратимых подсхем $\frS_1$--$\frS_5$ и количество используемых ими дополнительных входов,
    мы получаем следующие оценки для искомой обратимой схемы $\frS$:
    \begin{gather*}
        L(\frS) \sim 2^k + 2^{n-k} + p2^{s+1} + pn 2^{n-k} + n 2^{n-k}
            \sim 2^k + \frac{2^{n-k+1}}{s} + \frac{n2^n}{s}   \; , \\
        q(\frS) \sim 2^k + 2^{n-k} + p2^s + n2^{n-k} + n \sim 2^k + \frac{2^{n-k}}{s} + n 2^{n-k}  \; .
    \end{gather*}
    
    Следовательно, при $k = \lceil n \mathop / \phi(n) \rceil$ и $s = n - 2k$,
    где $\phi(n) \leqslant n \mathop / (\log_2n + \log_2 \psi(n))$ и $\psi(n)$~--- некоторые растущие функции,
    верны следующие соотношения:
    \begin{gather*}
        L(\frS) \sim 2^{\lceil n \mathop / \phi(n) \rceil} + \frac{2^{n+1}}{n(1-o(1))2^{\lceil n \mathop / \phi(n) \rceil}}
            + \frac{n2^n}{n(1-o(1))} \sim 2^n   \; , \\
        q(\frS) \sim 2^{\lceil n \mathop / \phi(n) \rceil} + \frac{2^n}{n(1-o(1))2^{\lceil n \mathop / \phi(n) \rceil}}
            + \frac{n 2^n}{2^{\lceil n \mathop / \phi(n) \rceil}} \sim \frac{n 2^n}{2^{\lceil n \mathop / \phi(n) \rceil}}  \; .
    \end{gather*}
    
    Поскольку мы описали алгоритм синтеза обратимой схемы $\frS$ для произвольного булева отображения $f$,
    то $L(n,q_0) \leqslant L(\frS) \sim 2^n$, где $q_0 \sim n2^{n - \lceil n \mathop / \phi(n) \rceil}$.
\end{proof}


\section*{Заключение}

В данной работе был рассмотрен вопрос сложности обратимых схем, состоящих из функциональных элементов NOT, CNOT и 2-CNOT.
Была изучена функция Шеннона сложности $L(n, q)$ обратимой схемы, реализующей какое-либо отображение $\ZZ_2^n \to \ZZ_2^n$ из
множества $F(n,q)$, как функции от $n$ и количества дополнительных входов схемы $q$.
Были доказаны нижние и верхние оценки для $L(n, q)$ для обратимых схем, использующих и не
использующих дополнительные входы.
Было показано, что использование дополнительной памяти в обратимых схемах, состоящих из функциональных элементов 
NOT, CNOT и 2-CNOT почти всегда позволяет снизить сложность обратимой схемы, чего нельзя утверждать в общем случае про схемы,
состоящие из классических необратимых функциональных элементов.

При решении задачи синтеза обратимой схемы, реализующей какое-либо отображение, приходится искать компромисс
между сложностью синтезированной схемы и количеством используемых дополнительных входов в схеме.
Направлением дальнейших исследований является более детальное изучение зависимости этих двух величин друг от друга.


\end{document}